\documentclass[USenglish,acmsmall]{acmart} 
\usepackage{amsmath}
\usepackage{amsthm}
\usepackage{url}
\usepackage{epsfig}
\usepackage{psfrag}
\usepackage{amssymb}
\usepackage{wrapfig}
\usepackage{xcolor}
\usepackage{xspace}

\usepackage[ruled,vlined,linesnumbered]{algorithm2e}

\usepackage[T1]{fontenc}
\usepackage[utf8]{inputenc}

\newcommand{\set}[2]{\{\, #1 \, ; \, #2 \,\}}
\newcommand{\sset}[1]{\{ #1\}}
\newcommand{\abs}[1]{ | #1 | }
\newcommand{\floor}[1]{\lfloor #1 \rfloor}

\newcommand{\xalg}{x^{\mathit{alg}}}

\newcommand{\xstarstar}{x^{\mathit{optaux}}}
\newcommand{\xstarstarstar}{x^{\mathit{optrel}}}

\newcommand{\xss}{\xstarstar}
\newcommand{\xsss}{\xstarstarstar}
\newcommand{\Csss}{A_c^{\mathit{optrel}}}
\newcommand{\Usss}{A_u^{\mathit{optrel}}}

\newcommand{\NSW}{\mathrm{NSW}}
\newcommand{\eps}{\varepsilon}

\newcommand{\ot}{\leftarrow}
\newcommand{\barv}{\bar{v}}
\newcommand{\baru}{\bar{u}}

\newcommand{\pEF}{$p$-EF1\xspace}

\newcommand{\assign}{\mathrel{:=}}
\newcommand{\argmax}{\mathop{\mathrm{argmax}}}

\newcommand{\remove}[1]{}

\title{On Fair Division for Indivisible Items}\titlenote{A preliminary version of this article appeared in FSTTCS 2018.}

\author{Bhaskar Ray Chaudhury}
\affiliation{\institution{MPI for Informatics, Saarland Informatics Campus} \country{Germany}}
\author{Yun Kuen Cheung}
\affiliation{\institution{Singapore University of Technology and Design} \country{Singapure}}
\author{Jugal Garg}
\affiliation{\department{Dept. of Industrial and Enterprise Systems Engineering} \institution{Univ.~of Illinois at Urbana-Champaign} \country{USA}}
\author{Naveen Garg}
\affiliation{\department{Department of Computer Science} \institution{IIT Delhi} \country{India}}
\author{Martin Hoefer}
\affiliation{\department{Institut f\"ur Informatik} \institution{Goethe-Universit\"at Frankfurt am Main} \country{Germany}}
\author{Kurt Mehlhorn}
\affiliation{\institution{MPI for Informatics, Saarland Informatics Campus} \country{Germany}}

\begin{document}

\begin{abstract}  We consider the task of assigning indivisible goods to a set of agents in a fair manner. Our notion of fairness is Nash social welfare, i.e., the goal is to maximize the geometric mean of the utilities of the agents. Each good comes in multiple items or copies, and the utility of an agent diminishes as it receives more items of the same good. The utility of a bundle of items for an agent is the sum of the utilities of the items in the bundle. Each agent has a utility cap beyond which he does not value additional items. We give a polynomial time approximation algorithm that maximizes Nash social welfare up to a factor of $e^{1/{e}} \approx 1.445$. The computed allocation approximates envy-freeness up to one item up to a factor of $2 + \eps$. For instances without caps, it is approximately  Pareto-optimal. We also show that the upper bounds on the optimal Nash social welfare introduced in~\cite{Cole-Gkatzelis} and~\cite{DBLP:journals/corr/BarmanMV17} have the same value. \end{abstract}

\maketitle


\section{Introduction}
We consider the task of dividing indivisible goods among a set of $n$ agents in a fair manner. More precisely, we consider the following scenario. We have $m$ distinct goods. Goods are available in several copies or items; there are $k_j$ items of good $j$. The agents have decreasing utilities for the different items of a good, i.e., for all $i$ and $j$
\[   u_{i,j,1} \ge u_{i,j,2} \ge \ldots \ge u_{i,j,k_j} .\]
An allocation assigns the items to the agents. 
For an allocation $x$, $x_i$ denotes the multi-set of items assigned to agent $i$, and $m(j,x_i)$ denotes the multiplicity of good $j$ in $x_i$. Of course, $\sum_i m(j,x_i) = k_j$ for all $j$. The total utility of bundle $x_i$ for agent $i$  is given by
\[             u_i(x_i) =      \sum_j \sum_{1 \le \ell \le m(j,x_i)} u_{i,j,\ell}.\]
Each agent has a utility cap $c_i$. The capped utility of bundle $x_i$ for agent $i$ is defined as
\[        \baru_i(x_i) = \min(c_i,u_i(x_i)).    \]
Our notion of fairness is \emph{Nash social welfare} ($\NSW$)~\cite{Nash}, i.e., the goal is to maximize the geometric mean 
\[                   \NSW(x) = \left(\prod_{1 \le i \le n}  \baru_i(x_i) \right)^{{1}/{n}} \]
of the capped utilities. All utilities and caps are assumed to be integers.  We give a polynomial-time approximation algorithm with approximation guarantee $e^{{1}/{e}} + \eps \approx 1.445 + \eps$ for any positive $\eps$. 

The problem has a long history. For divisible goods, maximizing Nash Social Welfare (NSW) for any set of valuation functions can be expressed via an Eisenberg-Gale program~\cite{EG59}. Notably, for \emph{additive valuations} ($c_i = \infty$ for each agent $i$ and $k_j = 1$ for each good $j$) this is equivalent to a Fisher market with identical budgets. In this way, maximizing NSW is achieved via the well-known fairness notion of competitive equilibrium with equal incomes (CEEI)~\cite{Moulin03}. 

For indivisible goods, the problem is NP-complete~\cite{DBLP:journals/aamas/NguyenNRR14} and APX-hard~\cite{DBLP:journals/ipl/Lee17}. Several constant-factor approximation algorithms are known for the case of additive valuations. They use different approaches. 

The first one was pioneered by Cole and Gkatzelis~\cite{Cole-Gkatzelis} and uses spending-restricted Fisher markets. Each agent comes with one unit of money to the market. Spending is restricted in the sense that no seller wants to earn more than one unit of money. If the price $p$ of a good is higher than one in equilibrium, only a fraction ${1}/{p}$ of the good is sold. Cole and Gkatzelis showed how to compute a spending restricted equilibrium in polynomial time and how to round its allocation to an integral allocation with good NSW. In the original paper they obtained an approximation ratio of $2 e^{{1}/{e}} \approx 2.889$. Subsequent work~\cite{CDGJMVY17} improved the ratio to $2$. The approximation ratio in~\cite{Cole-Gkatzelis} is shown against an upper bound on the optimal Nash social welfare which we refer to as CG-bound. In~\cite{CDGJMVY17}, an alternative bound is introduced and shown to have the same value as the CG-bound.

The second approach is via stable polynomials. Anari et al.~\cite{AnariGSS17}~obtained an approximation factor of $e$. 

The third approach, introduced by Barman et al.~\cite{DBLP:journals/corr/BarmanMV17}, is via  integral allocations that are envy-free up to one good. An allocation is envy-free up to one good if for any two agents $i$ and $k$ there is a good $j$ such that $u_i(x_k - j) \le u_i(x_i)$, i.e., after removal of one good from $k$'s bundle its utility for $i$ is no larger than the utility of $i$'s bundle for $i$. Caragiannis et al.~\cite{CaragiannisKMP016} have shown that an allocation maximizing NSW is Pareto-optimal and envy-free up to one good. For a price vector $p$ for the goods, the price $P(x_i)$ of a bundle is the sum of the prices of the goods in the bundle. An allocation is almost price-envy-free up to one good ($\eps$-$p$-EF1) if $P(x_k - j) \le (1 + \eps) P(x_i)$ for all agents $i$ and $k$ and some good $j$, where $\eps$ is an approximation parameter. An allocation is MBB (maximum bang per buck) if $j \in x_i$ implies $u_{ij}/p_j = \max_\ell u_{i\ell}/p_\ell$ for all $j$ and $i$. Barman et al.~\cite{DBLP:journals/corr/BarmanMV17} studied allocations that are MBB and almost price-envy-free up to one good. They showed that such allocations are almost envy-free up to one good\footnote{\label{footnoteBarman}Consider two bundles $x_k$ and $x_i$ and assume $P(x_k - j) \le (1 + \eps)P(x_i)$ for some $j \in x_k$. Let $\alpha_i = \max_\ell u_{i\ell}/p_\ell$. Then $u_i(x_k - j) = \sum_{\ell \in x_k - j} u_{i\ell} \le \alpha_i \sum_{\ell \in x_k - j} p_{\ell} \le (1 + \eps) \alpha_i \sum_{\ell \in x_i} p_\ell = (1 + \eps) \sum_{\ell \in x_i} u_{i\ell}$.} and approximate NSW up to a factor $e^{{1}/{e}} + \eps \approx 1.445 + \eps$. They also showed how to compute such an allocation in polynomial time. The approximation ratio in~\cite{DBLP:journals/corr/BarmanMV17} is shown against an upper bound on the optimal Nash social welfare which we refer to as BMV-bound. We show that it has the same value as the  CG-bound.

There are also constant-factor approximation algorithms beyond additive utilities. 

Garg et al.~\cite{GargHoeferMehlhornNashWelfare} studied budget-additive utilities ($k_j = 1$ for all goods $j$ and arbitrary $c_i$). They showed how to generalize the Fisher market approach and obtained an $2e^{{1}/{2e}} \approx 2.404$-approximation. 

Anari et al.~\cite{Anari:SODA2018} investigated multi-item concave utilities ($c_i = \infty$ for all $i$ and $k_j$ arbitrary). They generalized the Fisher market and the stable polynomial approach and obtained approximation factors of $2$ and $e^2$, respectively.

We show that the price-envy-free allocation approach can handle both generalizations simultaneously. We obtain an approximation ratio of $e^{{1}/{e}} + \eps \approx 1.445 + \eps$. The allocation computed by our algorithm guarantees $u_i(x_k - j) \le (2 + \eps)u_i(x_i)$ for any two agents $i$ and $k$, i.e., it approximates envy-freeness up to one item up to a factor of essentially two. For instances without utility caps, it is approximately Pareto-optimal\footnote{The algorithm rounds each non-zero utility to the next larger power of $r = 1 + \eps$. For instances without utility caps, it computes a Pareto-optimal allocation for the rounded utilities. It also computes a scaling factor $\alpha_i$ for each agent $i$ such that for each integral allocation $y$ of the goods $\sum_i u_i(y_i)/\alpha_i \le r \cdot \sum_i u_i(\xalg_i)/\alpha_i$. For $r = 1$, this would imply Pareto-optimality of $\xalg$. For instances without caps, it can efficiently compute a Pareto-optimal solution if all numbers are polynomially bounded. In contrast, even for identical agents with additive utilities and caps, computing a Pareto-optimal solution is strongly NP-hard via a standard reduction from 3-PARTITION.} The approach via price-envy-freeness does not only yield better approximation ratios and guarantees for individual agents, it is, in our opinion, also simpler to state and simpler to analyze. 

The paper is structured as follows. In Section~\ref{capped} we give the algorithm and analyze its approximation ratio (Section~\ref{approximation ratio}), guarantee to individual agents (Section~\ref{individual guarantees}), and running time (Section~\ref{running time}). In Section~\ref{analysis is tight} we  show that the analysis is essentially tight by establishing a lower bound of 1.44 on the approximation ratio of the algorithm. In Section~\ref{Certification} we discuss  certification of the approximation ratio
and in Section~\ref{EnvyFreeness} we show that for the multi-copy case and the capped case optimal allocations  are not necessarily envy-free up to one good. In Section~\ref{Large Markets} we obtain close-to-one approximation ratios for large markets, where for any agent the value of any good is only an $\epsilon/n$-fraction of the value of all goods. In Section~\ref{Equivalence} we show the equality of the CG- and BMV-bounds.

\section{Algorithm and Analysis}\label{capped}

Let us recall the setting. Items are indivisible. There are $n$ agents and $m$ goods. There are $k_j$ items or copies of good $j$. Let $M = \sum_j k_j$ be the total number of items. The agents have decreasing utilities for the different items of a good, i.e., for all $i$ and $j$
\[   u_{i,j,1} \ge u_{i,j,2} \ge \ldots \ge u_{i,j,k_j} .\]
For an allocation $x$, $x_i$ denotes the multi-set of items assigned to agent $i$, and $m(j,x_i)$ denotes the multiplicity of good $j$ in $x_i$. The total utility of bundle $x_i$ for agent $i$ is given by
\[             u_i(x_i) =           \sum_j \sum_{1 \le \ell \le m(j,x_i)} u_{i,j,\ell}.\]
Each agent has a utility cap $c_i$. The capped utility of bundle $x_i$ for agent $i$ is defined as
\[        \baru_i(x_i) = \min(c_i,u_i(x_i)).    \]
Following~\cite{GargHoeferMehlhornNashWelfare}, we assume w.l.o.g.~$u_{i,j,\ell} \le c_i$ for all $i$, $j$, and $\ell$. In the algorithm, we ensure this assumption by capping every $u_{i,*,*}$ at $c_i$. All utilities and caps are assumed to be integers. 

\subsection{A Reduction to Rounded Utilities and Caps}
Let $r \in (1,{3}/{2}]$. For every non-zero utility $u_{i,j,\ell}$ let $v_{i,j,\ell}$ be the next larger power of $r$. For zero utilities $v$ and $u$ agree. Similarly, for $c_i$ let $d_i$ be the next larger power of $r$. It is well-known that it suffices to solve the rounded problem with a good approximation guarantee.

\begin{lemma} Let $x$ approximate the NSW for the rounded problem up to a factor of $\gamma$. Then $x$ approximates the NSW for the original problem up to a factor $\gamma r$. \end{lemma}
\begin{proof} Let $x^*$ be an optimal allocation for the original problem.  We write $\NSW(x^*,u,c)$ for the Nash social welfare of the allocation $x^*$ with respect to utilities $u$ and caps $c$. Define $\NSW(x,u,c)$,
$\NSW(x^*,v,d)$, and  $\NSW(x,v,d)$ analogously. We need to bound ${\NSW(x^*,u,c)}/{\NSW(x,u,c)}$. Since $u \le v$ and $c \le d$ componentwise, $\NSW(x^*,u,c) \le \NSW(x^*,v,d)$. Since $x$ approximates the NSW for the rounded problem up to a factor $\gamma$, $\NSW(x^*,v,d) \le \gamma  \NSW(x,v,d)$. Since $v \le r u$ and $d \le r c$ componentwise,  $\NSW(x,v,d) \le r \NSW(x,u,c)$. Thus 
\[ \frac {\NSW(x^*,u,c)}{\NSW(x,u,c)} \le \frac{\gamma \NSW(x,v,d)}{{\NSW(x,v,d)}/{r}} = \gamma r.\]
\end{proof}

\subsection{The Algorithm} 

Barman et al.~\cite{DBLP:journals/corr/BarmanMV17} gave a highly elegant approximation algorithm for the case of a single copy per good and no utility caps. We generalize their approach. The algorithm uses an approximation parameter $\eps \in (0,{1}/{4}]$. Let $r = 1 + \eps$. The nonzero utilities are assumed to be powers of $r$. 

The algorithm maintains an integral assignment $x$, a price $p_j$ for each good, and an MBB-ratio\footnote{In the case of one copy per good, $\alpha_i = {u_{i,j}}/{p_j}$ whenever (the single copy of) good $j$ is assigned to $i$ and $\alpha_i \ge u_{i,\ell}/p_\ell$ for all goods $\ell$. Thus $\alpha_i$ is the maximum utility per unit of money (maximum bang per buck (MBB)) that agent $i$ can get.} $\alpha_i$ for each agent. Of course, $\sum_i m(j,x_i) = k_j$ for each good $j$. The prices, MBB-ratios, and multiplicity of goods in bundles are related through the following inequalities:
\begin{equation}\label{interval for alpha}     \frac{u_{i,j,m(j,x_i) +1}}{p_j} \le \alpha_i \le  \frac{u_{i,j,m(j,x_i)}}{p_j}, \end{equation}
i.e., if ${u_{i,j,\ell}}/{p_j} > \alpha_i$, then at least $\ell$ copies of $j$ are allocated to agent $i$ and if ${u_{i,j,\ell}}/{p_j} < \alpha_i$, then less than $\ell$ copies of $j$ are allocated to agent $i$. If no copy of good $j$ is assigned to $i$, the upper bound for $\alpha_i$ is infinity. If all copies of good $j$ are assigned to $i$, the lower bound for $\alpha_i$ is zero. Note that if $\alpha_i$ is equal to its upper bound in (\ref{interval for alpha}), we may take one copy of $j$ away from $i$ without violating the inequality as the upper bound becomes the new lower bound. Similarly, if $\alpha_i$ is equal to its lower bound in (\ref{interval for alpha}), we may assign an additional copy of $j$ to $i$ without violating the inequality as the lower bound becomes the new upper bound. 
Since (\ref{interval for alpha}) must hold for every good $j$, $\alpha_i$ must lie in the intersection of the intervals for the different goods $j$, i.e.,
\[   \max_j \frac{u_{i,j,m(j,x_i) +1}}{p_j} \le \alpha_i \le  \min_j \frac{u_{i,j,m(j,x_i)}}{p_j}. \]
The value of bundle $x_i$ for $i$ is given by\footnote{In the case of one copy per good, $P_i(x_i) = {u_i(x_i)}/{\alpha_i}= \sum_{j \in x_i} p_j$ is the total price of the goods in the bundle. We reuse the letter $P$ for the value of a bundle, although $P_i(x_i) = 1/\alpha_i \cdot \sum_j \sum_{1 \le \ell \le m(j,x_i)} u_{i,j,\ell}$ is no longer the total price of the goods in the bundle.}
\begin{equation}\label{price of a bundle} 
P_i(x_i) = \frac{u_i(x_i)}{\alpha_i} = \frac{1}{\alpha_i} \sum_j \sum_{1 \le \ell \le m(j,x_i)} u_{i,j,\ell}.
\end{equation}
Definitions (\ref{interval for alpha}) and (\ref{price of a bundle}) are inspired by Anari et al~\cite{Anari:SODA2018}. We say that $\alpha_i$ is equal to the upper bound for the pair $(i,j)$ if $\alpha_i$ is equal to its upper bound in (\ref{interval for alpha}) and that $\alpha_i$ is equal to the lower bound for the pair $(i,j)$ if $\alpha_i$ is equal to its lower bound in  (\ref{interval for alpha}).

An agent $i$ is \emph{capped} if $u_i(x_i) \ge c_i$ and is \emph{uncapped} otherwise.

\begin{algorithm}[!t]
\small
\caption{\label{alg:FPTAS} Approximate Nash Social Welfare for Multi Item Concave Utilities with Caps}
\DontPrintSemicolon
\SetKwInOut{Input}{Input}\SetKwInOut{Output}{Output}
\Input{Fair Division Problem given by utilities $u_{ij\ell}$, $i \le n$, $j \le m$, $\ell \le k_j$, utility caps $c_i$, and approximation parameter $\eps \in (0,1/4]$. Let $r = 1 + \eps$. Nonzero $u_{ij}$'s and $c_i$'s are powers of $r$.}
\Output{Price vector $p$ and $4\eps$-\pEF integral allocation $x$ }
\For{$i,j,\ell$}{$u_{i,j,\ell} \ot \min(c_i,u_{i,j,\ell})$\;}
\For{$j \in G$}{
\For{$\ell \in [k_j]$ in increasing order}{
assign the $\ell$-th copy of $j$ to $i_0 = \argmax_i u_{i,j,m(j,x_i) + 1}$;
}
Set $p_j \ot u_{i_0,j,m(j,x_{i_0}) }$, where $i_0$ is the agent to which the $k_j$-th copy of $j$ was assigned\;}
\For{$i \in A$}{$\alpha_i = 1$\;}
\While{true}{
\If{allocation $x$ is $\eps$-\pEF}{{\bf break} from the loop and terminate\;}
Let $i$ be a  least spending uncapped agent\;
Perform a BFS in the tight graph starting at $i$\;
\If{the BFS-search discovers an improving path starting in $i$, let $P =  (i=a_0,g_1,a_1,\ldots,g_h,a_{h})$ be a shortest such
path}{
Set $\ell \ot h$\; \label{begin of then}
\While{$\ell > 0$ and $P_{a_\ell}(x_{a_\ell} - g_\ell) > (1 + \eps) P_i(x_i)$}{
remove $g_\ell$ from $x_{a_\ell}$ and assign it to $a_{\ell - 1}$; $\ell \ot \ell - 1$\; \label{end of then}}
}
\Else{
Let $S$ be the set of goods and agents that can be reached from $i$ in the tight graph\;
$\beta_1  \ot \min_{k \in S;\ j \not\in S} {\alpha_k}/{({u_{k,j,m(j,x_k)+1}}/{p_j})}$\hfill (add a good to $S$)\;
$\beta_2 \ot \min_{k \not\in S;\ j \in S} {(u_{k,j,m(j,x_k)}/p_j)}/{\alpha_k}$   \hfill  (add an agent to $S$)\;
$\beta_3 \ot  \frac{1}{r^2 P_i(x_i)} \max_{k \not\in S} \min_{j \in x_k} P_k(x_k - j)$  \hfill ($i$ is happy)\;
$\beta_4 \ot r^s$, where $s$ is the smallest integer such that $r^{s-1} \le {P_h(x_h)}/{P_i(x_i)}< r^s$ and $h$ is the least spending uncapped agent outside $S$  \hfill (new least spender)\;
$\beta \ot \min(\beta_1,\beta_2,\max(1,\beta_3),\beta_4)$\;
multiply all prices of goods in $S$ by $\beta$ and divide all MBB-values of agents in $S$ by $\beta$\;
\If{$\beta_3 \le \min(\beta_1,\beta_2,\beta_4)$}{break from the while-loop\;}
}
}
\end{algorithm}

The algorithm starts with a greedy assignment. For each good $j$, it assigns each copy to the agent that values it most. The price of each good is set to the utility of the assignment of its last copy and all MBB-values are set to one. 
Note that this setting guarantees (\ref{interval for alpha}) for every pair $(i,j)$. Also, all initial prices and MBB-values are powers of $r$. It is an invariant of the algorithm that prices are powers of $r$. Only the final price increase in the main-loop may destroy this invariant. 

After initialization, the algorithm enters a loop. We need some more definitions. An agent $i$ is a \emph{least spending} uncapped agent if it is uncapped and $P_i(x_i) \le P_k(x_k)$ for every other uncapped agent $k$. 
An agent $i$ $\eps$-$p$-envies agent $k$ up to one item if $P_k(x_k - j) > (1 + \eps)\cdot P_i(x_i)$ for every good $j \in x_k$. Recall that $x_k$ is a multi-set. In the multi-set $x_k - j$, the number of copies of good $j$ is reduced by one, i.e., $m(j, x_k - j) = m(j,x_k) - 1$. Therefore $P_k(x_k - j) = P_k(x_k) - {u_{k,j,m(j,x_k)}}/{\alpha_k}$.
An allocation is \emph{$\eps$-$p$-envy free up to one item ($\eps$-$p$-EF1)} if for every uncapped agent $i$ and every other agent $k$ there is a good $j$ such that $P_k(x_k - j) \le (1 + \eps)P_i(x_i)$. 

We also need the notion of the \emph{tight graph}. It is a directed bipartite graph with the agents on one side and the goods on the other side. We have a directed edge $(i,j)$ from agent $i$ to good $j$ if 
$\alpha_{i} = {u_{ijm(j,x_i) + 1}}/{p_j}$, i.e., $\alpha_{i}$ is at its lower bound for the pair $(i,j)$. We have a directed edge $(j,i)$ from good $j$ to agent $i$ if $\alpha_{i} = {u_{ijm(j,x_i) }}/{p_j}$, i.e., $\alpha_{i}$ is at its upper bound for the pair $(i,j)$. Note that necessarily $m(j,x_i) \ge 1$ in the latter case, since otherwise good $j$ does not impose an upper bound for $\alpha_i$. 

An \emph{improving path} starting at an agent $i$ is a simple path $P = (i=a_0,g_1,a_1,\ldots,g_h,a_{h})$ in the tight graph starting at $i$ and ending at another agent $a_h$ such that $P_{a_h}(x_{a_h} - g_h) > (1 + \eps) P_i(x_i)$ and 
$P_{a_\ell}(x_{a_\ell} - g_\ell) \le (1 + \eps) P_i(x_i)$ for $1 \le \ell < h$. 

Let $i$ be the least spending uncapped agent. We perform a breadth-first search in the tight graph starting at $i$.
If the BFS discovers an improving path starting at $i$, we use the shortest such path to improve the allocation. Note that if $i$ $\eps$-$p$-envies some node that is reachable from $i$ in the tight graph then the BFS will discover an improving path. 

\begin{figure}
\begin{center}
\includegraphics[width=0.7\textwidth]{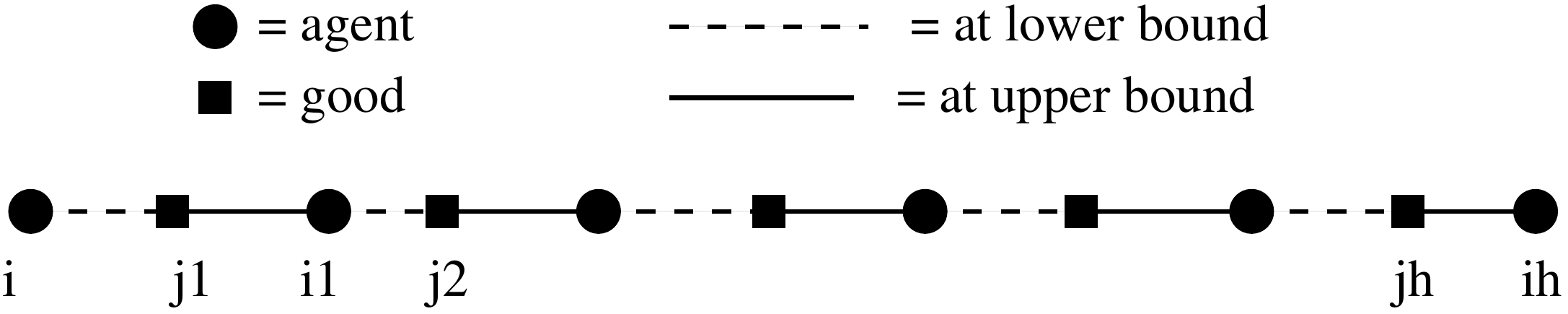}
\end{center}
\caption{\label{alternating path} An improving path. Agents and goods alternate on the path and the path starts and ends with an agent. For the solid edges $(j,i)$, $\alpha_i$ is at its upper bound for the pair $(i,j)$ and for the dashed edges $(i,j)$, $\alpha_i$ is at its lower bound for the pair $(i,j)$.}
\end{figure}

In the main loop, we distinguish cases according to whether BFS discovers an improving path starting at $i$ or not. 

Assume first that BFS discovers the improving path $P = (i=a_0,g_1,a_1,\ldots,g_h,a_{h})$. We take $g_h$ away from $a_h$ and assign it to $a_{h-1}$. If we now have
$P_{a_{h-1}}(x_{a_{h-1}} + g_h - g_{h-1}) \le  (1 + \eps) P_i(x_i)$ we stop. Otherwise, we take $g_{h-1}$ away from $a_{h-1}$ and assign it to $a_{h-2}$. If we now have $P_{a_{h-2}}(x_{a_{h-2}} + g_{h-1} - g_{h-2}) \le  (1 + \eps) P_i(x_i)$ we stop. Otherwise, \ldots. We continue in this way until we stop or assign $g_1$ to $a_0$. In other words, let $h' < h$ be maximum such that 
$P_{a_{h'}}(x_{a_{h'}} + g_{h'+1} - g_{h'}) \le (1 + \eps)P_i(x_i)$. If $h'$ exists, then we take
a copy of $g_\ell$ away from $a_\ell$ and assign it to $a_{\ell - 1}$ for $h' < \ell \le h$. If $h'$ does not exist, we do so for $1 \le \ell \le h$. Let us call the above a sequence of swaps. 

\begin{lemma}\label{effect of a sequence of swaps} Consider an execution of lines (\ref{begin of then}) to (\ref{end of then}) and let $h'$ be the final value of $\ell$ (this agrees with the definition of $h'$ in the preceding paragraph). Let $x'$ be the resulting allocation. Then $x'_\ell = x_\ell$ for $0 \le \ell < h'$, $x'_{h'} = x_{h'} + g_{h'+1}$, $x'_\ell = x_\ell + g_{\ell + 1} - g_\ell$ for $h' < \ell < h$, and $x'_h = x_h - g_h$. Also, 
\begin{itemize}
\item $P_{a_h}(x_{a_h}) \ge P_{a_h}(x'_{a_h}) > (1 + \eps) P_i(x_i)$, 
\item $P_{a_{h'}}(x'_{a_{h'}} - g_{h'}) = P_{a_{h'}}(x_{a_{h'}} + g_{h'+1} - g_{h'}) \le (1 + \eps)P_i(x_i)$ if $h' \ge 1$ 
\item $P_{a_0}(x'_{a_0}  - g_1) = P_{a_0}(x_{a_0}) \le (1 + \eps)P_i(x_i)$ if $h' = 0$. 
\item 
$P_{a_\ell}(x'_{a_\ell}) = P_{a_\ell}(x_{a_{\ell}} + g_{\ell+1} - g_{\ell}) > (1 + \eps)P_i(x_i)$ and $P_{a_\ell}(x'_{a_\ell} - g_{\ell+1}) = P_{a_\ell}(x_{a_\ell} - g_\ell) \le (1 + \eps)P_i(x_i)$ for $h' < \ell < h$. 
\item $P_{a_\ell}(x'_{a_\ell} - g_\ell) = P_{a_\ell}(x_{a_\ell} - g_\ell) \le (1 + \eps)P_i(x_i)$ for $0 \le \ell < h'$. 
\end{itemize} 
 \end{lemma}
\begin{proof} Immediate from the above.\end{proof}

If $i$ is still the least spending uncapped agent after an execution of lines (\ref{begin of then}) to (\ref{end of then}), we search for another improving path starting from $i$. We will show below that $i$ can stay the least spending agent for at most $n^2M$ iterations. Intuitively this holds because for any agent (factor $n$) and any fixed length shortest improving path (factor $n$), we can have at most $M$ iterations for which the shortest improving path ends in this particular agent. 

We come to the else-case, i.e., BFS does not discover an improving path starting at $i$. This implies that $i$ does not $\eps$-$p$-envy any agent that it can reach in the tight graph. We then increase some prices and decrease some MBB-values. Let $S$ be the set of agents and goods that can be reached from $i$ in the tight graph. 

\begin{lemma}\label{closure properties of $S$} If a good $j$ belongs to $S$ and $\alpha_k$ is at its upper bound for the pair $(k,j)$, then $k$ belongs to $S$. If an agent $k$ belongs to $S$ and $\alpha_k$ is at its lower bound for the pair $(k,j)$, then $j$ belongs to $S$. \end{lemma}
\begin{proof} Consider any good $j \in S$. Since $j$ belongs to $S$, there is an alternating path starting in $i$ and ending in $j$. If the path contains $k$, $k$ belongs to $S$. If the path does not contain $k$, we can extend the path by $k$. In either case, $k$ belongs to $S$. 

Consider any agent $k \in S$. Since $k$ belongs to $S$, there is an alternating path starting in $i$ and ending in $k$. If the path contains $j$, $j$ belongs to $S$. If the path does not contain $j$, we can extend the path by $j$. In either case, $j$ belongs to $S$. 
\end{proof}

We multiply all prices of goods in $S$ and divide all MBB-values of agents in $S$ by a common factor $t \ge 1$. What is the effect?
\begin{itemize}
\item Let ${u_{k,j,(j,x_k) +1}}/p_j \le \alpha_k \le {u_{k,j,m(j,x_k)}}/{p_j}$ be the inequality (\ref{interval for alpha}) for the pair $(k,j)$. The endpoints do not move if $j \not\in S$ and are divided by $t$ for $j \in S$. Similarly, $\alpha_k$ does not move if $k \not\in S$ and are divided by $t$ if $k \in S$. So in order to preserve the inequality, we must have: If $\alpha_k$ is equal to the upper endpoint and $p_j$ moves, i.e., $j \in S$, then $\alpha_k$ must also move. If $\alpha_k$ is equal to the lower endpoint and $\alpha_k$ moves then $p_j$ must also move. Both conditions are guaranteed by Lemma~\ref{closure properties of $S$}.
\item If $k$ and $j$ are both in $S$, then $\alpha_k$ and the endpoints of the interval for $(k,j)$ move in sync. So agents and goods reachable from $i$ in the tight graph, stay reachable. 
\item If $k \not\in S$, there might be a $j \in S$ such that $\alpha_k$ becomes equal to the right endpoint of the interval for $(k,j)$. Then $k$ is added to $S$. 
\item If $k \in S$, there might be a $j \not \in S$ such that $\alpha_k$ becomes equal to the left endpoint of the interval for $(k,j)$. Then $j$ is added to $S$. 
\item For agents in $S$, $P_k(x_k)$ is multiplied by $t$. For agents outside $S$, $P_k(x_k)$ stays unchanged. 
\end{itemize}

How is the common factor $t$ chosen? There are four limiting events. Either $S$ grows and this may happen by the addition of a good (factor $\beta_1$) or an agent (factor $\beta_2$); or $P_i(x_i)$ comes close to the largest value of $\min_{j \in x_k} P_k(x_k - j)$ for any other agent (factor $\beta_3$), or $P_i(x_i)$ becomes larger than $P_h(x_h)$ for some uncapped agent $h$ outside $S$ (factor $\beta_4$). Since we want prices to stay powers of $r$, $\beta_4$ is chosen as a power of $r$. The factor $\beta_3$ might be smaller than one. Since we never want to decrease prices, we take the maximum of $1$ and $\beta_3$.

\begin{lemma} Prices and MBB-values are powers of $r$, except maybe at termination. \end{lemma}
\begin{proof} This is true initially, since prices are utility values and utility values are assumed to be powers of $r$ and since MBB-values are equal to one. If prices and MBB-values are powers of $r$ before a price update, $\beta_1$, $\beta_2$, and $\beta_4$ are powers of $r$. Thus prices and MBB-values are after the price update, except maybe when the algorithm terminates. 
\end{proof}

We next show that the algorithm terminates with an allocation that is almost price-envy-free up to one item. 

\begin{lemma}\label{Four-eps-pEF} Assume $\eps \le {1}/{4}$. When the algorithm terminates, $x$ is a $4\eps$-$p$-EF1 allocation. \end{lemma}
\begin{proof}
Let $q$ be the  price vector after the price increase and let $h$ be the least spending uncapped agent after the increase; $h = i$ is possible. We first show that that $Q_i(x_i) \le r Q_h(x_h)$. This is certainly true if $h = i$. If $h \not\in S$, since the price increase is limited by $\beta_4$, we have
\[ Q_i(x_i) = \beta P_i(x_i) \le \beta_4 P_i(x_i) = r \cdot r^{s-1} \cdot P_i(x_i) \le r P_h(x_h) = r Q_h(x_h).\]
So in either case, we have $Q_i(x_i) \le r Q_h(x_h)$. Moreover, $Q_h(x_h) \le Q_i(x_i)$ because $h$ is a least spending uncapped agent after the price increase. 

If the algorithm terminates, we have $\beta_3 \le \beta_4$. Consider any agent $k$. Then, for $k \in S$, 
\begin{align*}
Q_k(x_k - j_k) &\le (1 + \eps) Q_i(x_i) \le (1 + \eps) \cdot r \cdot Q_h(x_h) \\
\intertext{and, for $k \not\in S$,}
Q_k(x_k - j_k) &= P_k(x_k - j_k) \le \beta_3 (1 + \eps) r P_i(x_i) = (1 + \eps) r Q_i(x_i)  \le (1 + \eps) \cdot r^2 \cdot Q_h(x_h).
\end{align*}
Thus we are returning an allocation that is $((1+\eps)r^2 - 1)$-$q$-EF1.
Finally, note that $(1 + \eps) r^2 = (1 + \eps)^3 \le (1 + 4 \eps)$ for $\eps \le {1}/{4}$. 
\end{proof}

\noindent {\bf Remark:} We want to point out the differences to the algorithm by Barman et al. Our definition of alternating path is more general than theirs since it needs to take into account that the number of items of a particular good assigned to an agent may change. For this reason, we need to maintain the MBB-ratio explicitly. In the algorithm by Barman et al.~the MBB ratio of agent $i$ is equal to the maximum utility to price ratio $\max _j {u_{ij}}/{p_j}$ and only MBB goods can be assigned to an agent. As a consequence, if a good belongs to $S$, the agent owning it also belongs to $S$. In price changes, there is no need for the quantity $\beta_2$. In the definition of $\beta_3$, we added an additional factor $r^2$ in the denominator. We cannot prove polynomial running time without this factor. Finally, we start the search for an improving path from the least uncapped agent and not from the least agent.

\subsection{Analysis of the Approximation Factor}\label{approximation ratio}

The analysis refines the analysis given by Barman et al. Let $(\xalg,p,\alpha)$ denote the allocation and price and MBB vector returned by the algorithm.  Recall that $\xalg$ is $\gamma$-$p$-EF1 with  $\gamma = 4\eps$ with respect to $p$ and (\ref{interval for alpha}) holds for every $i$. We scale all the utilities of agent $i$ and its utility cap by $\alpha_i$, i.e., we replace $u_{i,j,\ell}$ by ${u_{i,j,\ell}}/{\alpha_i}$ and $c_i$ by ${c_i}/{\alpha_i}$ and use $u_{i,j,\ell}$ and $c_i$  also for the scaled utilities and scaled utility cap. The scaling does not change the integral allocation maximizing Nash Social Welfare. Inequality (\ref{interval for alpha}) becomes 
\begin{equation}\label{final interval for alpha}     \frac{u_{i,j,m(j,\xalg_i) +1}}{p_j} \le 1 \le  \frac{u_{i,j,m(j,\xalg_i)}}{p_j}, \end{equation}
i.e., the items allocated to $i$ have a utility to price ratio of one or more and the items that are not allocated to $i$ have a ratio of one or less. Also, the value of bundle $x_i$ for $i$ is now equal to its utility for $i$ and is given by
\begin{equation}\label{final price of a bundle} P_i(\xalg_i) = u_i(\xalg_i) = \sum_j \sum_{1 \le \ell \le m(j,\xalg_i)} u_{i,j,\ell}.\end{equation}
All $u_{i,*,*}$ are at most $c_i$.

Let $A_c$ and $A_u$ be the set of capped and uncapped agents in $\xalg$, let $c = \abs{A_c}$ and $n - c = \abs{A_u}$ be their cardinalities. We number the uncapped agents such that $u_1(\xalg_1) \ge u_2(\xalg_2) \ge \ldots \ge u_{n-c}(\xalg_{n -c})$. Let $\ell = u_{n-c}(\xalg_{n-c})$ be the minimum utility of a bundle assigned to an uncapped agent. The capped agents are numbered $n - c + 1$ to $n$. Let $x^*$ be an integral allocation maximizing Nash social welfare. 

We define an auxiliary problem with $\sum_j k_j$ goods and one copy of each good. The goods are denoted by triples $(i,j,\ell)$, where $1 \le \ell \le m(j,\xalg_i)$. The utility of good $(i,j,\ell)$ is uniform for all agents and is equal to $u_{i,j,\ell}$. Formally,
\begin{equation}\label{uniform problem}     v_{*,(i,j,\ell)} = u_{i,j,\ell},  \end{equation}
where $v$ is the utility function for the auxiliary problem. The cap of agent $i$ is $c_i$. Since $v$ is uniform, we can write $v(x_i)$ instead of $v_i(x_i)$. The capped utility of $x_i$ for agent $i$ is $\barv_i(x_i) = \min(c_i, v(x_i))$. Note that $v$ is uniform, but $\barv$ is not. Let $\xstarstar$ be an optimal allocation for the auxiliary problem. 

\newcommand{\uorig}{u^{\mathrm{orig}}}

\begin{lemma}\label{Auxiliary Problem} Let $u_i$ be the scaled utilities. Then we have:
\begin{enumerate}
\item $\xalg$ maximizes the uncapped social welfare, i.e., $\xalg = \arg \max_x \sum_i u_i(x_i)$.
\item $\sum_i u_i(x^*_i) \le \sum_i u_i(\xalg_i) = \sum_{i,j,1 \le \ell \le m(j,\xalg_i)} v_{*,(i,j,\ell)}$.
\item $\xalg$ is Pareto-optimal for uncapped utilities.
\item $\NSW(x^*) = \left( \prod_{i} \baru_i(x_i^*)\right)^{1/n} \le  \left( \prod_{i} \barv_i(\xstarstar_i)\right)^{1/n} = \NSW(\xstarstar)$.
  \item Let $\uorig$ be the unrounded original utilities. Then for any integral allocation $y$ of the goods, $\sum_i \uorig_i(y_i)/\alpha_i \le r \cdot \sum_i \uorig_i(\xalg_i)/\alpha_i$. 
\end{enumerate}
\end{lemma}
\begin{proof} For part (1) consider $x^{SW}$ as the allocation that maximizes the uncapped social welfare for the scaled utilities. We can obtain $x^*_{SW}$ from $\xalg$ by moving copies of goods as follows:

Set $x \ot \xalg$. Consider any good $j$. As long as the multiplicities of $j$ in the bundles of $x$ and $x^{SW}$ are not the same, identify two agents $i$ and $k$, where $x_i$ contains more copies of $j$ than $x^{SW}_i$ and $x_k$ contains fewer copies of $j$ than $x^{SW}_k$, and move a copy of $j$ from $i$ to $k$. Each copy taken away has a utility of at least $p_j$, each copy assigned additionally has a utility of at most $p_j$. Thus the social welfare cannot go up by reassigning. This proves (a). 

Part (2) is an obvious consequence of part (1).

For part (3), we need to show that if agents have no utility caps, then there is no other allocation $y$ that satisfies $u_i(y_i) \ge u_i(\xalg_i)$ for all agents $i$, with strict inequality for at least one agent. This follows directly from part (1). Note that scaling does not affect the Pareto inequalities, thus part (3) also holds for unscaled uncapped utilities. 

For part (4), we interpret $\xalg$ as an allocation for the auxiliary problem; goods $(i,j,\ell)$ with $1 \le \ell \le m(j,\xalg_i)$ are allocated to agent $i$. We then move goods exactly as in (1). We obtain an allocation $\hat{x}$ for the auxiliary problem with $u_i(x^*_i) \le v(\hat{x}_i)$ for all $i$.

For part (5), we observe that $\uorig_i(y_i) \le u_i(y_i)$ since each original non-zero utility is scaled up to the next power of $r$, $\sum_i u_i(y_i)/\alpha_i \le \sum_i u_i(\xalg_i)/\alpha_i$ by part (c), and $u_i(\xalg_i) \le r \cdot \uorig_i(\xalg_i)$. 
\end{proof}

We stress that Lemma~\ref{Auxiliary Problem} refers to the scaled utilities; $\xalg$ does not maximize social welfare for the unscaled utilities.

For any agent $i$, let $b_i \in \xalg_i$ be such that $u_i(\xalg_i - b_i) \le (1 + \gamma) \ell$. Note that $u_i(\xalg_i - b_i) = u_i(\xalg_i) - u_{i,b_i,m(b_i,\xalg_i)}$. Let $B = \set{(i,b_i,m(b_i,\xalg_i))}{1 \le i \le n}$ be the goods in the auxiliary problem corresponding to the $b_i$'s. We now consider allocations for the auxiliary problem that are allowed to be partially fractional. We require that the goods in $B$ are allocated integrally and allow all other goods to be assigned fractionally. For convenience of notation, let $g_i = (i,b_i,m(b_i,\xalg_i))$. The following lemma is crucial for the analysis.

\newcommand{\xh}{\xstarstarstar}

\begin{lemma} There is an optimal allocation for the relaxed auxiliary problem in which good $g_i$ is allocated to agent $i$. \end{lemma}
\begin{proof} Assume otherwise. Among the allocations maximizing Nash social welfare for the relaxed auxiliary problem, let $\xh$ be the one that maximizes the number of agents $i$ that are allocated their own good $g_i$. 

Assume first that there is an agent $i$ to which no good in $B$ is allocated. Then $g_i$ is allocated to some agent $k$ different from $i$. Since $b_i \in \xalg_i$, $v(g_i) = u_{i,b_i,m(b_i,\xalg_i)} \le c_i$. The inequality holds since utilities $u_{i,*,*}$ are capped at $c_i$ during initialization. We move $g_i$ from $k$ to $i$ and $\min(v(g_i),v(\xh_i))$ value from $i$ to $k$. This is possible since only divisible goods are allocated to $i$. If we move $v(g_i)$ from $i$ to $k$, the NSW does not change. If $v(g_i)> v(\xh_i)$ and hence $c_i \ge v(g_i)> v(\xh_i)$, the product $\barv_i(x_i) \cdot \barv_k(x_k)$ changes from 
\begin{align*}  \min(c_i, v(\xh_i)) &\cdot \min(c_k, v(\xh_k - g_i + g_i)) \\ &= 
  \min(c_k v(\xh_i) , v(\xh_k - g_i)v(\xh_i) + v(g_i)v(\xh_i))
  \end{align*}
to
\begin{align*}
  \min(c_i, v(g_i)) &\cdot \min(c_k, v(\xh_k - g_i + \xh_i)) \\ &
                          \min(c_k v(g_i), v(\xh_k - v(g_i))v(g_i) + v(\xh_i)v(g_i)) .
\end{align*}
The arguments of the min in the lower line are componentwise larger than those of the min in the upper line. 
We have now modified $\xh$ such that the NSW did not decrease and the number of agents owning their own good increased. The above applies as long as there is an agent owning no good in $B$. 

So assume every agent $i$ owns a good in $B$, but not necessarily $g_i$. Let $i$ be such that $v(g_i)$ is largest among all goods $g_i$ that are not allocated to their $i$. Then $g_i$ is allocated to some agent $k$ different from $i$. The value of the good $g_\ell$ allocated to $i$ is at most $v(g_i)$ since $\ell \not= i$ and by the choice of $i$. We move $g_i$ from $k$ to $i$ and $\min(v(g_i),v(\xh_i))$ value from $i$ to $k$. This is possible since $v(g_\ell) \le v(g_i)$ and all other goods assigned to $i$ are divisible. We have now modified $\xh$ such that the NSW did not decrease and the number of agents owning their own good increased. We continue in this way until $g_i$ is allocated to $i$ for every $i$. \end{proof}

Let $\xsss$ be an optimal allocation for the relaxed auxiliary problem in which good $g_i$ is contained in the bundle $\xsss_i$ for every $i$. Let $\alpha$ be such that 
\[                     \alpha \ell = \min\set{v(\xsss_i)}{v(\xsss_i) < c_i} \]
is the minimum value of any agent that is uncapped in $\xsss$. Let $\alpha = \infty$, if every agent is capped in $\xsss$. Let $\Csss$ and $\Usss$ be the set of capped and uncapped agents in $\xsss$. Let $h$ be such that 
$u_h(\xalg_h) > \alpha \ell \ge u_{h+1}(\xalg_{h+1})$.




\begin{lemma}\label{upper bounds} For $i \le h$, $v(\xsss_i) \le u_i(\xalg_i)$. For all $i$, $u_i(\xalg_i) \le v(\xsss_i) +(1 + \gamma) \ell$. For $i \in A_u \cap \Csss$, $c_i \le \alpha \ell$ and $i \not\in [h]$. \end{lemma}
\begin{proof} Consider any $i \le h$. $v(\xsss_i) \le u_i(\xalg_i)$ is obvious, if $v(\xsss_i) \le \alpha \ell$. If $v(\xsss_i) > \alpha \ell$, then $\alpha < \infty$ and hence $\Usss$ is non-empty. We claim that $\xsss_i = \sset{g_i}$, i.e., $\xsss_i$ is a singleton consisting only of $g_i$. Assume otherwise, then also some divisible goods are assigned to $i$. We can move some of them to an agent that is uncapped in $\xsss$ and has value $\alpha\ell$. This increases the NSW, a contradiction. 

For the upper bound, we observe that $g_i \in \xsss_i$ and $u_i(\xalg_i - b_i) \le (1 + \gamma) \ell$. 

Consider next any $i \in A_u \cap \Csss$.  
Assume $c_i > \alpha \ell$. If $\xsss$ assigns divisible goods to $i$, we can move some of them to an agent that is uncapped in $\xsss$ and has value $\alpha \ell$. This increases the NSW. Thus $\xsss_i$ consists only of $g_i$. But then $v(g_i) \le u_i(\xalg_i) < c_i$ and $i$ does not belong to $\Csss$. This shows $c_i \le \alpha \ell$.
Then also $i \not\in [h]$ because otherwise $c_i < u_i(\xalg_i)$ and hence $i$ would be capped in $\xalg$. \end{proof}

\begin{lemma}
\[ \NSW(x^*) \le \NSW(\xsss) \le \left(  (\alpha \ell)^{n - c - h - \abs{A_u \cap \Csss}} \cdot \prod_{i \in A_c \cup (A_u \cap \Csss)} c_i \cdot \prod_{1 \le i \le h} u_i(\xalg_i) 
    \right)^{\frac{1}{n}}.\]
Moreover, $c_i \le \alpha \ell$ for any $i \in A_u \cap \Csss$. \end{lemma}
\begin{proof} If $v(\xsss_i) \not= \alpha \ell$ then either $i \in A_c$ or $i \in A_u \cap \Csss$ or $i \in A_u \setminus \Csss$. In the first case, $v(\xsss_i) \le c_i$. In the second case, $v(\xsss_i) = c_i \le \alpha \ell$ and $i \not\in [h]$ by Lemma~\ref{upper bounds}. In the third case, $v(\xsss_i) \le u_i(\xalg_i)$ for $i \le h$. So assume $i > h$. Then $v(g_i) \le  u_i(\xalg_i) \le \alpha \ell$ and hence all value in $v(\xsss_i)$ above $\alpha \ell$ would be by fractional goods. They could be reassigned for an increase in NSW. We conclude that for the agents $i \in A_u \setminus \Csss$ with $i > h$, we have $v(\xsss_i) = \alpha \ell$. \end{proof}

We next bound $\NSW(\xalg)$ from below. We consider assignments $x$ for the auxiliary problem that agree with $\xalg$ for the agents in $A_c \cup [h]$ and reassign the value $\sum_{i \in A_u - [h]} u_i(\xalg_i)$ fractionally. Note that for any $i \in A_u - [h]$, $\ell \le u_i(\xalg_i) \le \min(c_i,\alpha \ell)$. The former inequality follows from $i \in A_u$ and the latter inequality follows from the definition of $h$ and $i \in A_u$. We reallocate value so as to  move $u_i(x_i)$ towards the bounds $\ell$ and $\min(c_i,\alpha \ell)$. As long as there are two agents whose value is not at one of their bounds, we shift value from the smaller to the larger. This decreases NSW. We end when all but one agent have an extreme allocation, either $\ell$ or  $\min(c_i,\alpha \ell)$. One agent ends up with an allocation $\beta \ell$ with 
$\beta \in [1,\alpha]$. 

Let us introduce some more notation. Write $A_u \cap \Csss$ as $S \cup T$, where the agents $i \in T$ end up at $c_i$ and the agents in $S$ end up at $\ell$. Also let $s$ and $t$ be the number of agents in $A_u \setminus \Csss$ that end up at $\ell$ and $\alpha \ell$ respectively. Then
\[ \NSW(\xalg) \ge \left(    \prod_{i \in A_c} c_i \cdot \prod_{1 \le i \le h} u_i(\xalg_i) \cdot \ell^s \cdot (\alpha \ell)^t \cdot (\beta\ell) \cdot  \prod_{i \in T} c_i \cdot \ell^{\abs{S}}                         \right)^{1/n}.\]
Note that $n - c  - h = s + t + 1 + \abs{S} + \abs{T}$. Therefore
\[ \frac{\NSW(x^*)}{\NSW(\xalg)} \le \left(  \alpha^s \cdot \frac{\alpha}{\beta} \cdot \prod_{i \in S} \frac{c_i}{\ell}             \right)^{1/n} \le \left(\left(\frac{ s \alpha + \frac{\alpha}{\beta} + \sum_{i\in S} \frac{c_i}{\ell}}{s + 1 + \abs{S}}\right)^{s + 1 + \abs{S}}\right)^{1/n},\]
where we used the inequality between geometric mean and arithmetic mean for the second inequality. 

The total mass allocated by $\xsss$ to the agents in $A_u - [h]$ is $(s + t + 1)\alpha \ell + \sum_{i \in S \cup T} c_i$. The allocation $\xalg$ wastes up to $(1 + \gamma) \ell$ for each $i \in A_c \cup [h]$ and uses $s\ell + t \alpha \ell + \beta \ell + \sum_{i \in T} c_i + \abs{S} \ell$ on the agents in $A_u - [h]$. Therefore 
\[ (s + t + 1)\alpha \ell + \sum_{i \in S \cup T} c_i \le (\abs{A_c} + h) (1 + \gamma) \ell + s\ell + t \alpha \ell + \beta \ell + \sum_{i \in T} c_i + \abs{S} \ell\]
and hence after rearranging, dividing by $\ell$ and adding ${\alpha}/{\beta}$ on both sides
\begin{align*}
s \alpha + \frac{\alpha}{\beta} + \sum_{i \in S} \frac{c_i}{\ell} &\le (1 + \gamma) (\abs{A_c} + h) + s + \abs{S} + \frac{\alpha}{\beta} + \beta -\alpha
\\ &\le (1 + \gamma)(\abs{A_c} + h) + s + \abs{S} + 1 \le (1 + \gamma) n.
\end{align*}
Note that $\beta + {\alpha}/{\beta} - \alpha \le 1$ for $\beta \in [1,\alpha]$, since the expression is one at $\beta = 1$ and $\beta = \alpha$ and it second derivative as function of $\beta$ is positive. 
Thus
\begin{align*}
\frac{\NSW(\xsss)}{\NSW(\xalg)} \le \left(\left(\frac{(1 + \gamma)(\abs{A_c} + h) + s + \abs{S} + 1 }{s + 1 + \abs{S}}\right)^{s + 1 + \abs{S}}\right)^{1/n}
\le  \left(\frac{(1 + \gamma) n}{s + 1 + \abs{S}}\right)^{{(s + 1 + \abs{S})}/{n}} \le e^{e^{-{1}/{(1 + \gamma)}}},
\end{align*}
since $((1 + \gamma) \delta)^{{1}/{\delta}}$ as a function of $\delta$ attains its maximum for $\delta = \frac{1}{(1 + \gamma)} e^{{1}/{(1 + \gamma)}}$. The value of the maximum is
$\exp(\exp(-{1}/{(1 + \gamma)}))$. Table~\ref{function table} contains concrete values for small non-negative values of $\gamma$. 

\begin{table}[t]
\begin{center}
\begin{tabular}{|r|r|r|r|r|r|} \hline
$1 + \gamma$ & 1.00 & 1.01 & 1.02 & 1.03 & 1.04 \\ \hline
$\exp(\exp(-{1}/{(1 + \gamma)}))$ &
1.44467 & 1.44997 &1.45523 & 1.46046  & 1.46566\\ \hline
\end{tabular}
\end{center}
\caption{The factor $\exp(\exp(-{1}/{(1 + \gamma)}))$ as a function of $1 + \gamma$.}\label{function table}
\end{table}

\begin{theorem} Let $\eps \in (0,1/4]$, let $\gamma = 4 \eps$, let $\xalg$ be the allocation computed by the algorithm for the rounded utilities and caps, and let $x^*$ be an allocation maximizing Nash social welfare for the rounded utilities and caps. Then
  \[    \NSW(x^*)/\NSW(\xalg) \le e^{e^{-1/(1 +\gamma)}}.\]
\end{theorem}

\begin{corollary} $\eps \in (0,1/4]$, let $r = 1+ \eps$, let $\gamma = 4 \eps$, let $\xalg$ be the allocation computed by the algorithm for the utilities and caps rounded to powers of $r$, and let $x^*$ be an allocation maximizing Nash social welfare for the original utilities and caps. Then
  \[    \NSW(x^*)/\NSW(\xalg) \le r \cdot e^{e^{-1/(1 +\gamma)}}.\]
  \end{corollary}

\subsection{Guarantees for Individual Agents}\label{individual guarantees}

The allocation computed by our algorithm 
maximizes NSW up to a factor 1.45. By Lemma~\ref{Four-eps-pEF}, it also gives any uncapped agent $i$ the guarantee $\min_{j \in x_k} P_k(x_k - j) \le (1 + 4\eps) P_i(x_i)$ for every other agent $k$. This guarantee is not meaningful for agent $i$ as the left hand side is in terms of the utility for agent $k$. We now show that it implies $\min_{j \in x_k} u_i(x_k - j) \le (2 + 4\eps) u_i(x_i)$, i.e., the utility for $i$ of $k$'s bundle minus one item is essentially bounded by twice the utility of $i$'s bundle for $i$. The proof shows that the additional utility for $i$ of the items that $k$ has in excess of $i$ up to one item is bounded by $(1 + \eps) u_i(x_i)$. In the case of one copy per good, $x_k$ and $x_i$ are disjoint and hence any item in $x_k$ is in excess of $i$'s possession of the same good.

\begin{theorem}\label{guarantee for an agent}  The allocation computed by the algorithm satisfies $\min_{j \in x_k} u_i(x_k - j) \le (2 + {4}\eps) u_i(x_i)$ for any agent $k$ and any uncapped agent $i$. 
\end{theorem}
\begin{proof} Let $g$ be such that $P_k(x_k - g) \le (1+ 4\eps) P_i(x_i)$. Then
  \begingroup
  \allowdisplaybreaks
  \begin{align*}
   u_i(x_k - g) & \le u_i (x_i \cup x_k - g) & \text{more never harms}\\
                                  &= u_i(x_i) + \sum_j \sum_{\ell = m(j,x_i) + 1}^{m(j,x_k \cup x_i - g)} u_{i,j,\ell}\\
                                  &\le u_i(x_i) +  \sum_j \sum_{\ell = m(j,x_i) + 1}^{m(j,x_k \cup x_i - g)} \alpha_i p_j &\text{since $u_{i,j,\ell}/p_j \le \alpha_i$ for $\ell > m(j,x_i)$}\\
                                  &\le u_i(x_i) +  \sum_j \sum_{\ell = 1}^{m(j,x_k - g)} \alpha_i p_j \\
                                  &\le u_i(x_i) + \sum_j \sum_{\ell = 1}^{m(j,x_k - g)} \alpha_i \frac{u_{k,j,\ell}}{\alpha_k} &\text{since $u_{k,j,\ell}/p_j \ge \alpha_k$ for $k \le m(j,x_k)$}\\
                &\le u_i(x_i) + \alpha_i P_k(x_k - g) &\text{definition of $P_k(x_k - g)$}\\
                &\le u_i(x) + \alpha_i {(1+4\eps)} P_i(x_i)   &\text{since $P_k(x_k - g) \le {(1+4\eps)} P_i (x_i)$}\\
    &= (2 + {4}\eps) u_i(x_i) &\text{since $u_i(x_i) = \alpha_i P_i(x_i)$.}
  \end{align*}
  \endgroup
\end{proof}

For the case of only a single copy per good, $\min_{j \in x_k} u_k(x_k - j) \le (1 + \eps) u_i(x_i)$ holds as was shown by Barman et al.~\cite{DBLP:journals/corr/BarmanMV17}; see Footnote~\ref{footnoteBarman} for a proof. We do not know whether the factor 2 in the Theorem above is best possible. We show in Section~\ref{EnvyFreeness} that a factor larger than 1.2 is necessary.

\subsection{Polynomial Running Time}\label{running time}

Recall that $n$ is the number of agents, $m$ is the number of goods, there are $k_j$ copies of good $j$, and $M = \sum_j k_j$ is the total items. We also define $U = \max_{i,j,k} u_{ijk}/\min_{i,j,k: u_{ijk} \neq 0}$ as the ratio of the maximum to minimum non-zero utility.

The analysis follows Barman et al.~with one difference. Lemma~\ref{max price} is new. For its proof, we need the revised definition of $\beta_3$.

\begin{lemma} The price of the least spending uncapped agent is non-decreasing. \end{lemma}
\begin{proof} This is clear for price increases. Consider a sequence of swaps along an improving  path $P = (i=a_0,g_1,a_1,\ldots,g_h, a_h)$, where the agent $a_h$ loses a good, the agents $a_\ell$, $h' < \ell < h$, lose and gain a good, and the agent $a_{h'}$ gains a good. By Lemma 1, all agents $a_\ell$ with $h' < \ell \le h$ have a price of at least $(1 + \eps) P_i(x_i)$ after the swap. Also the price of agent $a_{h'}$ does not decrease. \end{proof}

\begin{lemma}\label{max price} For any agent $k$, let $j_k$ be a highest price item in $x_k$. Then $\max_k P_k(x_k - j_k)$ does not increase in the course of the algorithm as long as this value is above $(1 + \eps) \min_{\text{uncapped $i$}} P_i(x_i)$. Once $\max_k P_k(x_k - j_k) \le (1 + \eps) \min_{\text{uncapped $i$}}P_i(x_i)$, the algorithm terminates. \end{lemma}
\begin{proof} We first consider price increases and then a sequence of swaps. 

Consider any price increase which is not the last. Then $\beta_4 \le \beta_3$. Let $h$ be the least uncapped spender after the price increase and $q$ be the price vector after the increase. Then $Q_h(x_h) \le Q_i(x_i) \le r Q_h(x_h)$. For $k \in S$, we have $min_j Q_k(x_k - j) \le (1 + \eps) Q_i(x_i) \le (1 + \eps) r Q_h(x_h)$, i.e., agents in $S$ can become violators but we can bound how bad they can become. For the agent $k \not\in S$ defining $\beta_3$, we have 
\[ \min_j P_k(x_\ell - j) = \beta_3 (1 + \eps) r P_i(x_i) \ge (1 + \eps) r Q_i(x_i) \ge (1 + \eps)r Q_h(x_h) \]
and hence the worst violator stays outside $S$. We used the equality $r = 1 + \eps$ and the inequality 
$Q_i(x_i) = \beta P_i(x_i) \le \beta_3 P_i(x_i)$ in this derivation. 

Consider next a sequence of swaps. We have an improving path from $i$ to $k$, say $P = (i=a_0,g_1,a_1,\ldots,g_h, a_h = k)$. Let $x'$ be the allocation after the sequence of swaps. Then $\min_j P_k(x'_k - j) \le \min_j P_k(x_k - j)$ since $k$ loses a good and $\min_j P_\ell(x'_\ell - j) \le (1 + \eps)P_i(x_i)$ for all $\ell \in [0,h-1]$ by Lemma~\ref{effect of a sequence of swaps}. 
\end{proof}

\begin{lemma} The number of subsequent iterations with no change of the least spending agent and no price increase is bounded by $n^2 M$. \end{lemma}
\begin{proof} Let $i$ be the least spending agent. We count for any other agent $k$, how often the improving path can end in $k$. For each fixed length of the improving path, this can happen at most $M$ times (for details see~\cite{DBLP:journals/corr/BarmanMV17}). The argument is similar to the argument used in the strongly polynomial algorithms for weighted matchings~\cite{Edmonds-Karp}. 
\end{proof}

\begin{lemma} If the least spending uncapped agent changes after a price increase, the value of the old least spending uncapped agent increases by a factor of at least $r$. \end{lemma}
\begin{proof} The least uncapped spender changes if $\beta = \beta_4$ and $\beta_4$ is at least $r$. So $P_i(x_i)$ increases by at least $r$. \end{proof}


\begin{theorem} The number of iterations is bounded by $n^3 M^2 \log_r MU$. \end{theorem}
\begin{proof} Divide the execution into two parts. In the first part, there are agents that own no good, and in the second part every agent owns at least one good and hence all the $P_i(x_i)$ are non-zero.
  
  In any iteration of the first part $P_i(x_i) = 0$, where $i$ is a least spending agent. An shortest improving path $P =  (i=a_0,g_1,a_1,\ldots,g_h,a_{h})$ starting in $i$ visits agents $a_1$ to $a_{h-1}$ owning exactly one good and ends in agent $a_h$ owning more than one good. The sequence of swaps will take away $g_h$ from $a_h$ and assign $g_{i+1}$ to $a_i$ for $0 \le i < h$. Since every price increase will grow $S$ by either a good or an agent, an improving path will exist after at most $n + m$ iterations. Thus there are only $O(n^2)$ iterations in the first part. 

We come to the second part. Divide its execution into maximum subsequences with the same least spender. Consider any fixed agent $i$ and the subsequences where $i$ is the least spender. At the end of each subsequence, $i$ receives an additional item, or we have a price increase. In the latter case, $P_i(x_i)$ is multiplied by at least $r$. Consider the subsequences between price increases. At the end of a subsequence $i$ receives an additional item. It may or may not keep this item until the beginning of the next subsequence. If there are more than $M$ subsequences with $i$ being the least spender, there must be two subsequences such that $i$ loses an item between these subsequences. According to Lemma~\ref{effect of a sequence of swaps}, the value of $i$ after the swap is at least $r$ times the minimum price of any bundle and hence at least $r$ times the price of bundle $i$ when $i$ was least spender for the last time. Thus $P_i(x_i)$ increases by a factor of at least $r$. 

  We have now shown: After at most $M \cdot n^2 M$ iterations with $i$ being the least spender, $P_i(x_i)$ is multiplied by a factor $r$. Thus there can be at most $n^2 M^2 \log_r MU$ such iterations. Multiplication by $n$ yields the bound on the number of iterations. \end{proof}

\section{A Lower Bound on the Approximation Ratio of the Algorithm}\label{analysis is tight}

We show that the performance of the algorithm is no better than $1.44$. Let $k$, $s$ and $K$ be positive integers with $K \ge k$ which we fix later. Consider the following instance. We have $h = s(k-1)$ goods of value $K$ and $n = h + s$ goods of value $1$. There is one copy of each good. The number of agents is $n$ and all agents value the goods in the same way. 

The algorithm may construct the following allocation. There are $h$ agents that are allocated a good of value $1$ and a good of value $K$ and there are $s$ agents that are allocated a good of value $1$. This allocation can be constructed during initialization. The prices are set to the values and the algorithm terminates. 

The optimal allocation will allocate a good of value $K$ to $h$ players and spread the $h + s = sk$ goods of value $1$ across the remaining $s$ agents. So $s$ agents get value $k$ each. 
Thus
\[ \frac{\NSW(OPT)}{\NSW(ALG)} = \left(\frac{ K^h k^s}{(K + 1)^h}\right)^{{1}/{(h  + s)}} = \left(\left(\frac{K}{K+1}\right)^{(k-1)s} k^s \right)^{{1}/{ks}} =\left( \frac{K}{K+1}\right)^{{(k-1)}/{k}} k^{{1}/{k}}.\]

\noindent 
The term involving $K$ is always less than one. It approaches $1$ as $K$ goes to infinity. The second term $k^{{1}/{k}}$ has it maximal value at $k = e$. However, we are restricted to integral values. We have $2^{{1}/{2}} = 1.41$ and $3^{{1}/{3}} = 1.442$. 
For $k = 3$, $({K}/{(K+1)})^{{2}/{3}} = \exp(\frac{2}{3} \ln (1 - 1/{(K+1)})) \approx \exp(- \frac{2}{3(K+1)}) \approx 1 - \frac{2}{3(K+1)}$. So for $K = 666$, the factor is less than $1 - 1/{1000}$ and therefore ${\NSW(OPT)}/{\NSW(ALG)} \ge 1.440$. 

\section{Certification of the Approximation Ratio}\label{Certification}

How can a user of an implementation of the algorithm be convinced that the solution returned has a $\NSW$ no more than $1.445$ times the optimum? She may read this paper and convince herself that the program indeed implements the algorithm described in this article. This is unsatisfactory~\cite{CertifyingAlgs}. In this section, we describe an alternative certificate. 

The algorithm returns an allocation $\xalg$, prices $p_j$ for the goods, and MBB-ratios $\alpha_i$ for the agents. After scaling all utilities and the utility gap of agent $i$ by $\alpha_i$, we have (\ref{final interval for alpha}). The user needs to understand that this scaling has no effect on the optimal allocation. As in Section~\ref{approximation ratio}, we introduce the auxiliary problem with $M = \sum_j k_j$ goods and one copy of each good. The goods have uniform utilities. The user needs to understand that the $\NSW$ of the auxiliary problem is an upper bound (Lemma~\ref{Auxiliary Problem}). We are left with the task of convincing the user of an upper bound on the $\NSW$ of the auxiliary problem. 

 \begin{figure}[th]
\begin{center}\includegraphics[width=0.7\textwidth]{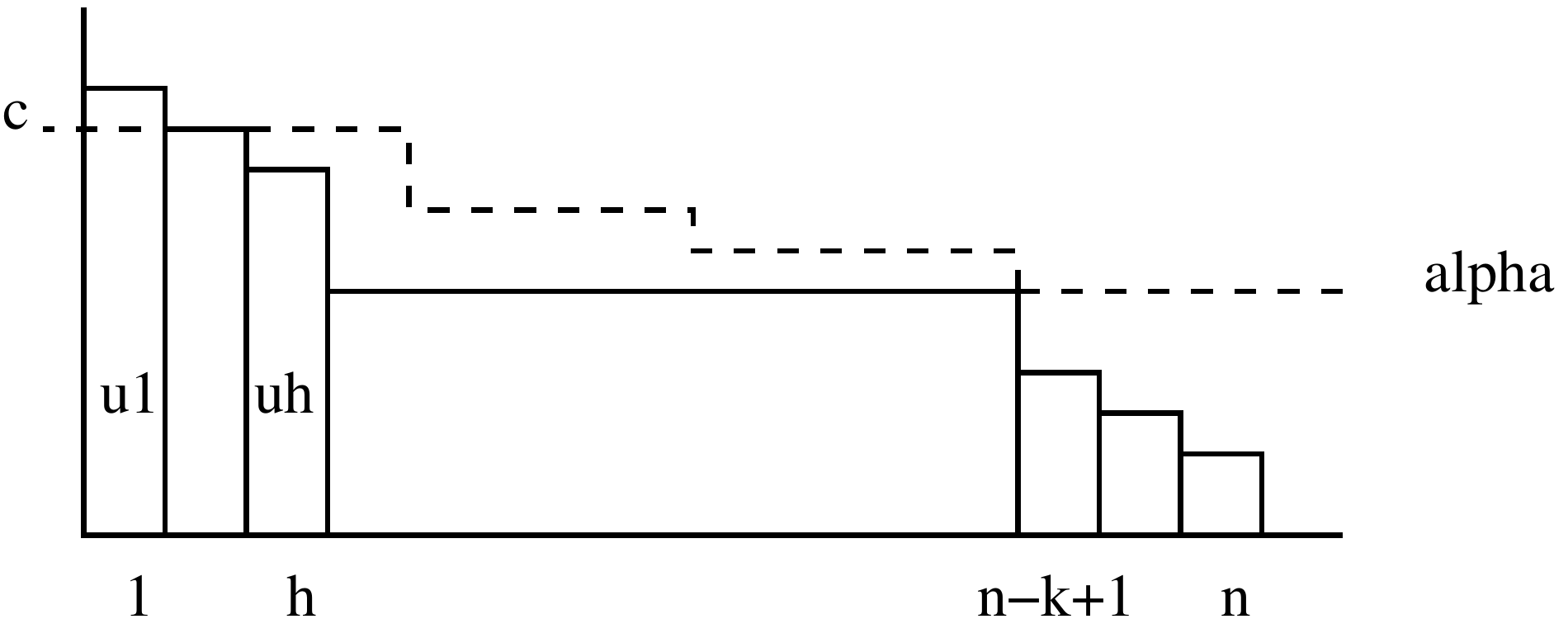}\end{center}
\caption{The allocation constructed in the proof of Theorem~\ref{optimal allocation}. The dashed line above agents $1$ to $n - k$ indicates the utility caps. The solid rectangles visualize the values of the bundles. }\label{auxiliary problem}
\end{figure}

\begin{theorem}\label{optimal allocation} Let $c_1 \ge c_2 \ge \ldots \ge c_n$ be the utility caps of the agents, let $u_1 \ge u_2 \ge \ldots \ge u_M$ be the utilities of the $M$ goods of the auxiliary problem, and let $\xss$ be an optimal allocation for the auxiliary problem. Then 
\[     \NSW(\xss) \le \left( \prod_{1 \le i \le h} \min(c_i,u_i) \cdot \delta^{n - h - k} \cdot \prod_{n - k + 1 \le i \le n} c_i \right)^{1/n}, \]
where $\delta = {\left( \sum_{h + 1 \le j \le M} u_j - \sum_{n - k + 1 \le i \le n} c_i \right)}/{(n - h - k)}$ and $h$ and $k$ are such that $h < n - k$ and 
$c_{n - k +1 } \le \delta < c_{n-k}$ and $\delta < u_h$. The right hand side is illustrated in Figure~\ref{auxiliary problem}. 
\end{theorem}
\begin{proof} We insist that the goods $1$ to $h$ are allocated integrally and allow the remaining goods to be allocated fractionally. Clearly, we cannot allocate more than $c_i$ to any agent, in particular, not to agents $n - k +1$ to $n$ and to agents $1$ to $h$. The optimal way to distribute value $\sum_{h + 1 \le j \le M} u_j$ to agents $h+1$ to $n$ is clearly to allocate $\delta$ each to agents $h+1$ to $n - k$ which all have a cap of more than $\delta$ and to the assign their cap to agents $n - k +1$ to $n$. The items $u_1$ to $u_h$ of value more than $\delta$ are best assigned to the agents with the largest utility caps. Assume that two such items, say $u_\ell$ and $u_k$, are allocated to the same agent. Then one of the first $h$ agents is allocated no such item; let $v$ be the value allocated to this agent. Moving $u_k$ to this agent and value $\min(u_k,v)$ from this agent in return, does not decrease the $\NSW$. Also, if any fractional items are assigned in addition to the first $h$ agents, we move them to agents $h+1$ to $n-k$
and increase the $\NSW$. This establishes the upper bound. \end{proof}

The upper bound can be computed in time $O(n^2 + M)$. We conjecture that it can be computed in linear time $O(n + M)$. We also conjecture that the bound is never worse than the bound used in the analysis of the algorithm. It can be better as the following example shows. We have two uncapped agents and three goods of value $u_1 = 3$, $u_2 = 1$ and $u_3 = 1$, respectively. The algorithm may assign the first two goods to the first agent and the third good to the second agent. The set $B$ in the analysis of the algorithm consists of the first good and the last good. 
Then $\ell = 1$. The optimal allocation allocates $3$ to the first agent and $2$ to the second agent. Thus $\alpha \ell = 2$. The analysis uses the upper bound $\sqrt{4 \cdot 2}$ for the $\NSW$ of the optimal allocation. The theorem above gives the upper bound $\sqrt{3 \cdot 2}$; note that $h=1$, $k = 0$, and $\delta = 2$. 

\section{Envy-Freeness up to one Copy}\label{EnvyFreeness}

For the case of additive valuations and one copy of each good, the optimal allocation is envy-free up to one good as shown in~\cite{CaragiannisKMP016}. Also the allocation constructed by the algorithm by Barman et al.~\cite{DBLP:journals/corr/BarmanMV17} is envy-free up to one good. In this section, we show that these properties hold neither for the multi-copy case nor for the capped case. 

Let $\eps$ be a small positive real, say $\epsilon = 0.01$, let $r = 1 + \eps$, and let $s$ be the smallest power of $r$ greater or equal to $2r^2$. Then $2 < s < 2.04$.
We first give an example for the multi-copy uncapped case. There are two agents and two goods. 
Good $1$ has 5 copies, and good $2$ has 2 copies.
For the first agent, the utility vector for good $1$ is $(s,s,0,0,0)$ and for good $2$ is $(1,0)$. For the second agent, the utility vector for good $1$ is $(s,s,s,0,0)$ and for good $2$ is $(s,s)$.
Then at the optimal NSW allocation, the first agent is allocated two copies of good $1$ and none of good $2$,
while the second agent is allocated three copies of good $1$ and two copies of good $2$. For this allocation, $\NSW = (2s \cdot 5s)^{1/2} = 10^{1/2} s$. Note that allocating one copy each of the second good to each agent gives a $\NSW$ of $((2s + 1) \cdot 4s)^{1/2} = (8s^2 + 4s)^{1/2} < (10s^2)^{1/2}$ since $4 < 2s$. The 
Clearly, the first agent envies the second agent even after removing one copy (of either good) from the allocation of the second agent because $u_1(x_2 - g) \ge (2s + 1) > 2s = u_1(x_1)$ for any choice of $g$.

\begin{lemma} In the case of several copies per good, the allocation maximizing $\NSW$ is not necessarily envy-free up to one copy. \end{lemma}

What does the algorithm do? The initial assignment is equal to the optimal assignment and sets $p_1 = p_2 = s$ and $ \alpha_1 = \alpha_2 = 1$. Agent 1 is the least spending uncapped agent. The allocation is not $\eps$-$p$-EF1, since $P_1(x_1) = 2s$ and $P_2(x_2) = 5s$ and $\min_{g \in x_2} P_2(x_2 - g) = 4s$. The constraints on $\alpha_1$ are $[0,1]$ by the first good and $[1/s,1]$ by the second good. The tight graph consists only of agent 1. We enter the else-case of the main loop with $S = {1}$. Then $\beta_1 = s >2$, $\beta_2 = \infty$, $\beta_3 = 4s/(2s\cdot r^2) = 2/r^2 < 2$ and $\beta_4 = r^{1 + \floor{\log_r 5/2}} \ge 5/2 \ge \beta_3$. Thus $\beta = \beta_3$. We decrease $\alpha_1$to $r^2/2 \approx 1/2$ and terminate. Now $P_1(x_1) = (2/r^2) \cdot 2s = 4s/r^2$ and hence $(1 + 4\eps) P_1(x_1) \ge 4s = P_2(x_2 - g) = 4s$. The optimal allocation is now $4\eps$-$p$-envy free up to one copy.

We turn to possible improvements of Theorem~\ref{guarantee for an agent}. Since $\min_{g \in x_2} u_1(x_2 - g) = 2s + 1$ and $u_1(x_1) = 2s$, in order to have $\min_{g \in x_2} u_1(x_2 - g) \le (c + \eps) u_1(x_1)$, we need $c \ge 1 + (1 - 2 \eps s)/(2 s) \ge 1.2$. 

\begin{lemma} With $\alpha_1 = r^2/2$, $\alpha_2 = s$, $p_1 = p_2$, the optimal allocation in the example above is $4\eps$-$p$-envy free up to one copy.\end{lemma}

\begin{lemma} Theorem~\ref{guarantee for an agent} does not hold when the constant 2 is replaced by 1.2. \end{lemma}

For the linear capped case, again we have two agents, and this time we have four goods with one copy each.
The utility vectors of both agents are $(s,s,s,s)$, but the first agent is capped at $3$, while the second agent is uncapped. Then the optimal NSW allocation allocates one good to the first agent and three goods to the second agent for $\NSW = (s \cdot 3s)^{1/2}$. Note that allocating 2 copies each give a $\NSW = (3 \cdot 2s)^{1/2} < (3 s^2)^{1/2}$ since $s \ge 2$. In the optimal assignment, the first agent envies the second agent, even after removing one good from the allocation of the second agent.

What does the algorithm do? It may construct the optimal assignment during initialization; the prices of all four goods are set to $s$ and both $\alpha$-values are set to one. Agent 1 is the least spending uncapped agent. The tight graph consists of the edges from agent 1 to the goods owned by agent 2 and from these goods to agent 1. An improving path exists and one of these goods is reassigned to agent 1. The algorithm terminates with an allocation in which both agents own two goods.

\begin{lemma} In the case of single copies per good but with utility caps, the allocation maximizing $\NSW$ is not necessarily envy-free up to one good. \end{lemma}

\section{Large Markets}\label{Large Markets}

Let $\delta > 0$ be a constant. We call a market $\delta$-large if $u_{i,j,\ell} \le \delta \cdot u_i(G)/n$, where $G$ is the set of all goods. Note that $u_i(G) = \sum_j \sum_{1 \le \ell \le k_j} u_{i,j,\ell}$. For simplicity, we restrict to instances without utility caps. With utility caps, the treatment becomes more clumsy, but does not give additional insights.

\begin{theorem} For a $\delta$-large market in which all non-zero utilities are powers of $r = 1 + \eps$
  \[      \NSW(x^*)/\NSW(\xalg)  \le (1 + 4 \eps)/(1 - \delta).   \]
\end{theorem}
\begin{proof} Let $(\xalg, p, \alpha)$ be the allocation, price vector, and scaling factors returned by the algorithm. For simplicity we use $x = \xalg$. We scale all utilities $u_{i,*,*}$ by $\alpha_i$. Then $\alpha_i$ becomes one and we have
  \[        u_{i,j,m(j,x_i) + 1} \le p_j \le u_{i,j,m(j,x_i)} \]
  for all $i$ and $j$. Let $U = \sum_i u_i(x_i) = \sum_i \sum_j \sum_{1 \le \ell \le m(j,x_i)} u_{i,j,\ell}$.

  Let $x^*$ be the allocation maximizing $\NSW$. Then $\sum_i u_i(x^*_i) \le U$ by Lemma~\ref{Auxiliary Problem}(a) and hence $\NSW(x^*) \le \left((U/n)^n\right)^{1/n} = U/n$.

  We next prove a lower bound on $\NSW(\xalg)$. Note that $u_i(G) \le U$ due to Lemma~\ref{Auxiliary Problem}(a). Thus, $u_{i,j,\ell} \le \delta U/n$.

  For any $i$, we have $P_i(x_i) = u_i(x_i)$. Since the allocation returned by the algorithm is $4\eps$-$p$-envy-free up to one copy, we have $\min_{g \in x_k} u_k(x_k - g) \le (1 + 4 \eps) u_i(x_i)$ for every agent $k$. Let $g_k$ be the good that minimizes the left hand side. Summing over all $k$ yields
  \[    \sum_k u_k(x_k) - \sum_k u_k(g_k) \le (1 + 4 \eps) n \cdot u_i(x_i),\]
  and hence
  \[   u_i(x_i) \ge \frac{U - n (\delta/n) U}{(1 + 4 \eps) n} = \frac{1 - \delta}{1 + 4 \eps} \cdot \frac{U}{n} .\]
  Thus
  \[
  \frac{\NSW(x^*)}{\NSW(\xalg)} \le \frac{U/n}{\left(\prod_i u_i(x_i)\right)^{1/n}} \le \frac{U/n}{\left(\left(\frac{1 - \delta}{1 + 4 \eps} \cdot \frac{U}{n}\right)^n\right)^{1/n}} = \frac{1 + 4 \eps}{1 - \delta}\enspace.
  \]
  \end{proof}

\section{The CG- and BMV-bound are Equal}\label{Equivalence}
\newcommand{\bu}{\overline{u}}
\newcommand{\lS}{S_\ell}
\newcommand{\sS}{S_s}
\newcommand{\CGbound}{\mathrm{CG\raisebox{1pt}{\hbox{-}}UB}}
\newcommand{\BMVbound}{\mathrm{BMV\raisebox{1pt}{\hbox{-}}UB}}

In this section, we restrict the discussion to the case of a single copy per good and no utility constraints. We have $n$ agents and $m$ goods. Cole and Gkatzelis~\cite{Cole-Gkatzelis} and Barman, Murthy, and Vaish~\cite{DBLP:journals/corr/BarmanMV17} defined upper bounds on the Nash social welfare of any integral allocation of the goods. We show that the bounds are equivalent. 

Cole and Gkatzelis~\cite{Cole-Gkatzelis} defined their upper bound via spending restricted Fisher markets. 
Each agent has one unit of money and each good has one unit of supply. Goods can be allocated fractionally and $x_{ij}$ is the fraction of good $j$ allocated to agent $i$. A solution to the market is an allocation $x$ and a price $p_j$ for each good $j$ such that 
\begin{enumerate}
\item Each agent spends all his money i.e., $\sum_j x_{ij}p_j = 1$.
\item An agent $i$ spends money only on goods with maximum bang-per-buck i.e., $x_{ij} > 0$ implies $u_{ij}/p_j = \alpha_i$ where $\alpha_i = \max_\ell u_{i\ell}/p_\ell$.
\item Goods with price less than 1 (small goods) are sold completely. Let $\sS=\{j|p_j< 1\}$. Then for all $j\in\sS$, $\sum_i x_{ij}=1$.
\item Exactly one unit of money is spent on each good with price at least 1 (large good). Let $\lS=\{j|p_j\ge 1\}$. Then for all $j\in\lS$, $\sum_i x_{ij}p_j=1$.
\end{enumerate}
The last constraint is the spending constraint and gives the market its name. Cole and Gkatzelis~\cite{Cole-Gkatzelis} show that the Nash social welfare of any integral allocation of goods to agents is at most

\[         \CGbound \assign  \left( \prod_{j \in \lS} p_j  \prod_i \alpha_i \right)^{1/n}. \]

\newcommand{\calS}{{\mathcal S}}

The following bound is implicit in the work of Barman, Murthy, and Vaish~\cite{DBLP:journals/corr/BarmanMV17}. For any scaling vector $\alpha = (\alpha_1,\ldots,\alpha_n)$, define uniform utilities $u_j$ by $u_j = \max_i u_{ij}/\alpha_i$. For a set $S$ of more than $m - n$ goods, let
\[                     a(S) = \frac{\sum_{j \in S} u_j}{\abs{S} - (m - n)} = \frac{\sum_{j \in S} u_j}{n - \abs{\overline{S}}};\]
note that $\abs{S} - (m - n) = n - (m - \abs{S}) = n - \abs{\overline{S}}$. So $a(S)$ is the amount per agent if the total utility of the goods in $S$ is distributed uniformly over $n -\abs{\overline{S}}$ agents. Finally, let
\[     \calS_\alpha = \set{S}{\abs{S} > m - n \text{ and } u_j > a(S) \text{ for $j \not\in S$}}.\]
Then the BMV-bound is defined as follows:
\[ \BMVbound \assign \min_{\alpha > 0} \min_{S \in \calS_\alpha}  \left( \prod_{j \not\in S} u_j \cdot a(S)^{n - \abs{\overline{S}}} \cdot \prod_i \alpha_i\right)^{1/n}. \]

\begin{lemma}$ \BMVbound$ is an upper bound on the Nash social welfare of any integral allocation. \end{lemma}
\begin{proof} Scaling the utilities of agent $i$ by $\alpha_i$ does not change the optimal allocation and changes the Nash social welfare of any allocation by $(\prod_i \alpha_i)^{1/n}$. Replacing $u_{ij}$ by $u_j = \max_h u_{hj}$ for every agent $i$ can only increase Nash social welfare. Allowing to allocate the goods in $S$ fractionally can only increase social welfare. Since $S$ is such that $u_j > a(S)$ for $j \not\in S$, the optimal partially fractional allocation is to allocate each $u_j$, $j \not\in S$, to a distinct agent and to allocate $a(S)$ to each one of the remaining agents. \end{proof}

\begin{lemma}\label{characterization of Salpha} For fixed $\alpha$, the BMV-bound is minimized for $S \in \calS_\alpha$ satisfying $u_j > a(S)$ for $j \not\in S$ and $u_j \le a(S)$ for $j \in S$. This $S$ is unique.\end{lemma}
\begin{proof} Assume $u_h > a(S)$ for some $h \in S$. Let $T = S - h$. Then
  \[   \left(u_h a(T)^{n - \abs{\overline{T}}}\right)^{1/(n - \abs{\overline{S}})} < a(S), \]
  since the LHS is the geometric mean of $u_h$ and $n - \abs{\overline{T}}$ copies of $a(T)$ and the RHS is
  equal to their arithmetic mean; note that $(u_h + (n - \abs{\overline{T}})a(T))/(n - \abs{\overline{S}}) = a(S)$.

  We can determine $S$ greedily. Start with $S$ equal to the set of all goods. As long as there is a $j \in S$ such that $u_j > a(S)$, remove $j$ from $S$. For\footnote{Let $k = n - \abs{\overline{S}}$. Then $a(S \setminus j) = (u(S) - u_j)/(k - 1) \le  u(S)/k = a(S)$ iff $u(S) \le k u_j$ iff $a(S) \le u_j$.} such a $j$, $a(S \setminus j) \le a(S)$ and hence any candidate for removal stays a candidate for removal. So the removal process always ends up with the same $S$. Also note that for $\abs{\overline{S}} = n - 1$, $a(S) = \sum_{j \in S} u_j \ge u_j$ for all $j \in S$ and hence the process stops before all goods are removed from $S$.
  \end{proof}

\begin{lemma} $\BMVbound \le \CGbound$. \end{lemma}
\begin{proof} Consider a solution $(x_{ij}, p_j)$ to the spending restricted Fisher market. The scaling vector for the $\BMVbound$ is now defined as $\alpha_i=\max_j u_{ij}/p_j$.
Let $\bu_{ij} = u_{ij}/\alpha_i$ be the scaled utilities. Then $\bu_{ij} \le p_j$ and $\bu_{ij} = p_j$ whenever $x_{ij} > 0$. Let $\bu_j = \max_i \bu_{ij}$; then $\bu_j = p_j$ since for every $j$, $x_{ij} > 0$ for at least one $i$. Since the total money spent is $n$, one unit is spent on each good in $\lS$, and $p_j$ is spent on good $j \in \sS$, we have 
\[      n = \sum_{j \in \sS} p_j + \abs{\lS}\]
and hence 
\[  a(\sS) = \frac{\bu(\sS)}{n - \abs{\overline{\sS}}} = \frac{\sum_{j \in \sS} p_j}{n - \abs{\lS}} = 1.\]
Since $p_j > 1$ for $j \in \lS$ and $p_j \le 1$ for $j \in \sS$, the $\BMVbound$ is minimized for the set $\sS\in S_{\alpha}$ and hence 
\[   \BMVbound \le  \left( \prod_{j \in \lS} p_j \cdot a(\sS)^{n - \abs{\lS}} \cdot \prod_i \alpha_i \right)^{1/n} =   \left( \prod_{j \in \lS} p_j  \prod_i \alpha_i \right)^{1/n} = \CGbound.\]
\end{proof}

\newcommand{\alphaBMV}{\alpha^{\mathrm{BMV}}}
\newcommand{\SBMV}{S^{\mathrm{BMV}}}
\newcommand{\uBMV}{u^{\mathrm{BMV}}}

\newcommand{\alphaS}{\alpha^S}
\newcommand{\uS}{u^S}

For a set $S$ of more than $m - n$ goods, let $P_S$ be the following minimization problem in variables $\alpha_i$ and $u_j$.
\begin{align*}      \text{minimize}&& f_S(\alpha,u) \assign \sum_{j \not\in S} \ln u_j + &(n - \abs{\overline{S}}) \ln a(S) + \sum_i \ln \alpha_i\\
  \text{subject to}&&    u_j &\ge u_{ij}/\alpha_i &&\text{for all $i$ and $j$}\\
                   && u_j &\ge a(S) &&\text{for $j \not\in S$}
\end{align*}

If $P_S$ is feasible, let $b_S$ be the optimum objective value and let $(\alphaS,\uS)$ be an optimum solution. If $S$ is the set of all goods, $\alpha_i = 1$ for all $i$ and $u_j = \max_i u_{ij}$ is feasible solution. Let $S^*$ be such that (1) $P_{S^*}$ is feasible, (2) $b_{S^*}$ is minimum, and (3) among the $S$ satisfying (1) and (2),  $S$ has largest cardinality. 

\begin{lemma}\label{Sstar is Salphastar} For $S = S^*$, $\uS_j > a(S) $ for $j \not\in S$, $\uS_j \le a(S)$ for $j \in S$, and $\uS_j = \max_i u_{ij}/\alpha_i$.  \end{lemma}
\begin{proof} Assume first that $\uS_h > a(S)$ for some $h \in S$. Consider $T = S - h$. Then $(\alphaS,\uS)$ is a feasible solution of $P_T$ and $b_T < b_S$ by the proof of Lemma~\ref{characterization of Salpha}. 

  Assume next that $\uS_h = a(S)$ for some $h \not\in S$. Let $T = S \cup h$. Then $(\alphaS,\uS)$ is a feasible solution for $P_T$ and $b_T = b_S$, a contradiction to the choice of $S$.

  Assume $\uS_j > \max_i u_{ij}/\alpha_i$ for some $j$. Since $\uS_j > a(S)$ if $j \not\in S$, we may decrease $\uS_j$, staying feasible and decreasing the objective. 
\end{proof}

\begin{lemma} Let $S = S^*$. Then $(\alphaS,S)$ defines the BMV-bound. \end{lemma}
\begin{proof} Let $(\alphaBMV,\SBMV)$ define the BMV-bound and let $\uBMV_j = \max_i u_{ij}/\alphaBMV_i$ for all $j$.  Then $(\alphaBMV,\uBMV)$ is a feasible solution of $P_{\SBMV}$ and
  \[ f_{\SBMV}(\alphaBMV,\uBMV) = \frac{1}{n} \ln \BMVbound.\]
  Therefore
  \[ b_{S^*} \le b_{\SBMV} \le \frac{1}{n} \ln \BMVbound. \]

  Conversely, let $S = S^*$, let $(\alphaS,\uS)$ be an optimal solution to $P_S$, and let $S_{\alphaS}$ be the set minimizing the BMV-bound for $\alphaS$  Then 
$S_{\alphaS} = S^*$ by Lemma~\ref{Sstar is Salphastar} and hence
  \[       \frac{1}{n} \ln \BMVbound \le b_{S^*}.\]
  \end{proof}

\begin{lemma} $\CGbound \le \BMVbound$. \end{lemma}
\begin{proof}
Let $S = S^*$ and let $(\alphaS,\uS)$ be an optimal solution of problem $P_S$. We have shown above that $(\alphaS,S)$ defines the BMV-bound, and that $\uS_j > a(S) $ for $j \not\in S$, $\uS_j \le a(S)$ for $j \in S$, and $\uS_j = \max_i u_{ij}/\alpha_i$. Let $k = n - \abs{\overline{S}}$. 

The KKT conditions are necessary conditions for the optimum. Let $z_{ij}\ge 0$ for all $i$ and $j$, and $y_j \ge 0$ for $j \not\in S$ be the multipliers. Then we need to have (write $u_j \alpha_i \ge u_{ij}$ for the inequalities) 
\begin{align*}      1/u^S_j &= \sum_i  z_{ij}\alpha^S_i  + y_j &&\text{for $j \not\in S$}\\
  \frac{1}{a(S)} &= \sum_i  z_{ij}\alpha^S_i &&\text{for $j \in S$}\\
  1/\alpha^S_i &= \sum_j z_{ij} u^S_j   &&\text{for all $i$}\\
  z_{ij} > 0 &\Rightarrow \alpha^S_i u^S_j = u_{ij} &&\text{for all $i$ and $j$}\\
  y_j > 0 &\Rightarrow u^S_j = a(S)    &&\text{for all $j \not\in S$}.
\end{align*}

\noindent
Define $p_j = u^S_j/a(S)$ and $x_{ij} = a(S) z_{ij} \alpha^S_i $ and call $p_j$ the price of good $j$ and $x_{ij}$ the fraction of good $j$ allocated to agent $i$. The bang-per-buck ratio of agent $i$ is
\[ \alpha_i = \max_j \frac{u_{ij}}{p_j} = \frac{\alphaS_i \uS_j}{\uS_j/a(S)} = a(S) \alphaS_i.\]
We now rewrite and interpret the optimality conditions. 
\begin{itemize}
\item Since $u^S_j > a(S)$ for $j \not\in S$, $y_j = 0$ for $j \not\in S$.
\item The first condition becomes $1 = \sum_i x_{ij} p_j$, i.e., exactly one unit of money is spent on each good $j \not\in S$. Note that $p_j > 1$ for such goods. \remove{, $\sum_i x_{ij} \le 1$, i.e., goods $j \not\in S$ are not overallocated.}
\item The third condition becomes $1 = \sum_j x_{ij} p_j$, i.e., every agent spends exactly one unit of money.
\item The second condition becomes $\sum_i x_{ij} = 1$ for all $j \in S$, i.e. goods in $S$ are completely allocated, but not overallocated.
\item $x_{ij}>0$ implies $z_{ij}>0$ which in turn implies that $\alpha^S_i u^S_j = u_{ij}$. Hence $u_{ij}/p_j=a(S)\alpha^S_i=\alpha_i$ which means that good $j$ is allocated to $i$ only if it has the maximum bang-per-buck ratio.
\end{itemize}
This shows that the pair $(x,p)$ is a solution to the spending restricted Fisher market and the CG-bound for this solution is
\remove{satisfies the constraints of the restricted Fisher-market. We next show that it is the optimal solution. We have 
\[  u_i(x_i) = \sum_j u_{ij} x_{ij} = \sum_j a(S) u_{ij} z_{ij} \alpha^S_i = (\alpha^S_i)^2 a(S) \sum_j \uS_j z_{ij} = a(S) \alpha^S_i.\]
The next to last equality holds because of the fourth condition. Let $y$ be any other allocation satisfying the constraints of the restricted Fisher market. Then
\[ u_i(y_i) = \sum_j u_{ij} y_{ij} \le \sum_j \uS_j \alpha^S_i y_{ij} = \alpha^S_i \sum_j \uS_j y_{ij} = \alpha^S_i \sum_j p_j y_{ij} \le \uS \alpha^S_i.\]
and hence $(x,p)$ is the optimal solution to the restricted Fisher market.

We are now ready to compute the CG-bound. A good $j \not\in S$ has a price $p_j = \uS_j/a(S) > 1$ and a good  $j \in S$ has a price $p_j = \uS_j/a(S) \le 1$. 
Thus}
\[ \CGbound \le  \prod_{j \not\in S} \frac{\uS_j}{a(S)} \cdot \prod_i a(S) \alphaS_i = \prod_{j \not\in S} \uS_j \cdot a(S)^{n - \abs{\overline{S}}} \cdot \prod_i \alphaS_i = \BMVbound.\]
\end{proof}

We have now shown the main theorem of this section.

\begin{theorem} The CG-bound and the BMV-bound have the same value. \end{theorem}

In~\cite[Lemma 8]{CDGJMVY17} yet another mathematical program is given for the same bound. We include it for completeness.
\begin{align*}      \text{maximize }&& &\hspace{-2em}\frac{\prod_i \prod_j u_{ij}^{b_{ij}}}{\prod_j q_j^{q_j}} \\
  \text{subject to}&&   \sum_i b_{ij} &= q_j &&\text{for all $j$}\\
                   && \sum_{j} b_{ij}& = 1 &&\text{for $i$}\\
&& q_j &\le 1&&\text{for all $j$}\\
&&   b_{ij} & \ge 0 && \text{for all $i$ and $j$}.
\end{align*}
In the optimal solution to this program $b_{ij}$ is the amount of money spent by agent $i$ on good $j$, and $q_{j}$ is the total spending on good $j$ in the spending restricted market.

\paragraph*{Acknowledgement:} We want to thank Hannaneh Akrami for a careful reading of the paper.

\renewcommand{\htmladdnormallink}[2]{#1}

\end{document}